\documentclass[english,nolineno]{socg-lipics-v2021}

\graphicspath{{figs/}}

\usepackage{algorithm}
\usepackage[noend]{algpseudocode}

\title{On the Edge Crossings of the Greedy Spanner}

\author{David Eppstein}
    {Department of Computer Science, University of California, Irvine}
    {eppstein@uci.edu}
    {}
    {}

\author{Hadi Khodabandeh}
    {Department of Computer Science, University of California, Irvine}
    {khodabah@uci.edu}
    {https://orcid.org/0000-0003-3850-6739}
    {}

\funding{This work was supported in part by the US National Science Foundation under grant CCF-1616248.}

\authorrunning{D. Eppstein and H. Khodabandeh}
\Copyright{David Eppstein and Hadi Khodabandeh}



\EventEditors{Kevin Buchin and \'{E}ric Colin de Verdi\`{e}re}
\EventNoEds{2}
\EventLongTitle{37th International Symposium on Computational Geometry (SoCG 2021)}
\EventShortTitle{SoCG 2021}
\EventAcronym{SoCG}
\EventYear{2021}
\EventDate{June 7--11, 2021}
\EventLocation{Buffalo, NY, USA}
\EventLogo{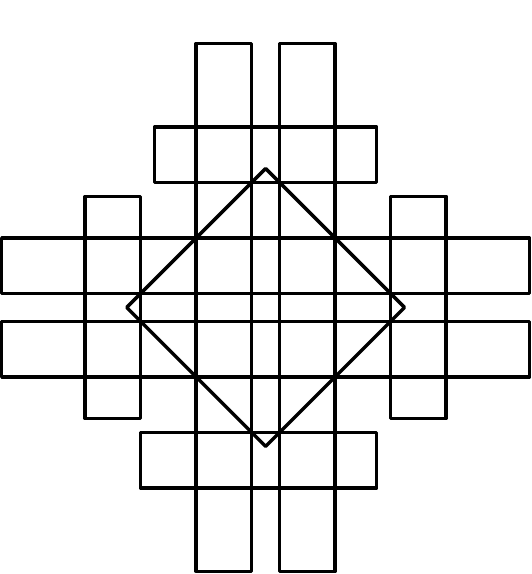}
\SeriesVolume{189}

\begin{document}

\maketitle

%\nolinenumbers
\hideLIPIcs

\begin{abstract}
The greedy $t$-spanner of a set of points in the plane is an undirected graph constructed by considering pairs of points in order by distance, and connecting a pair by an edge when there does not already exist a path connecting that pair with length at most $t$ times the Euclidean distance. We prove that, for any $t>1$, these graphs have at most a linear number of crossings, and more strongly that the intersection graph of edges in a greedy $t$-spanner has bounded degeneracy. As a consequence, we prove a separator theorem for greedy spanners: any $k$-vertex subgraph of a greedy spanner can be partitioned into sub-subgraphs of size a constant fraction smaller, by the removal of $O(\sqrt k)$ vertices. A recursive separator hierarchy for these graphs can be constructed from their planarizations in linear time, or in near-linear time if the planarization is unknown.\end{abstract}

\section{Introduction}
\label{sec:intro}
\emph{Geometric spanners} are geometric graphs whose distances approximate distances in complete graphs, while having fewer edges than complete graphs. Given a set of points $V$ on the Euclidean plane (or in any other metric space), a $t$-spanner on $V$ can be defined as a graph $S$ having $V$ as its set of vertices $V$ and satisfying the following inequality for every pair of points $(P,Q)$:
\begin{equation}d_S(P,Q)\leq t\cdot d(P,Q)\label{eq:stretch}\end{equation}
where $d_S(P,Q)$ is the length of the shortest path between $P$ and $Q$ using the edges in $S$, and $d(P,Q)$ is the Euclidean distance of $P$ and $Q$. We call \autoref{eq:stretch} the \emph{bounded stretch property}. Because of this inequality, $t$-spanners provide a $t$-approximation for the pairwise distances between the set of points in $V$. The parameter $t>1$ is called the \emph{stretch factor} or \emph{spanning ratio} of the spanner and determines how accurate the approximate distances are; spanners having smaller stretch factors are more accurate.

\begin{figure}[t]
\centering
\includegraphics[width=0.45\textwidth]{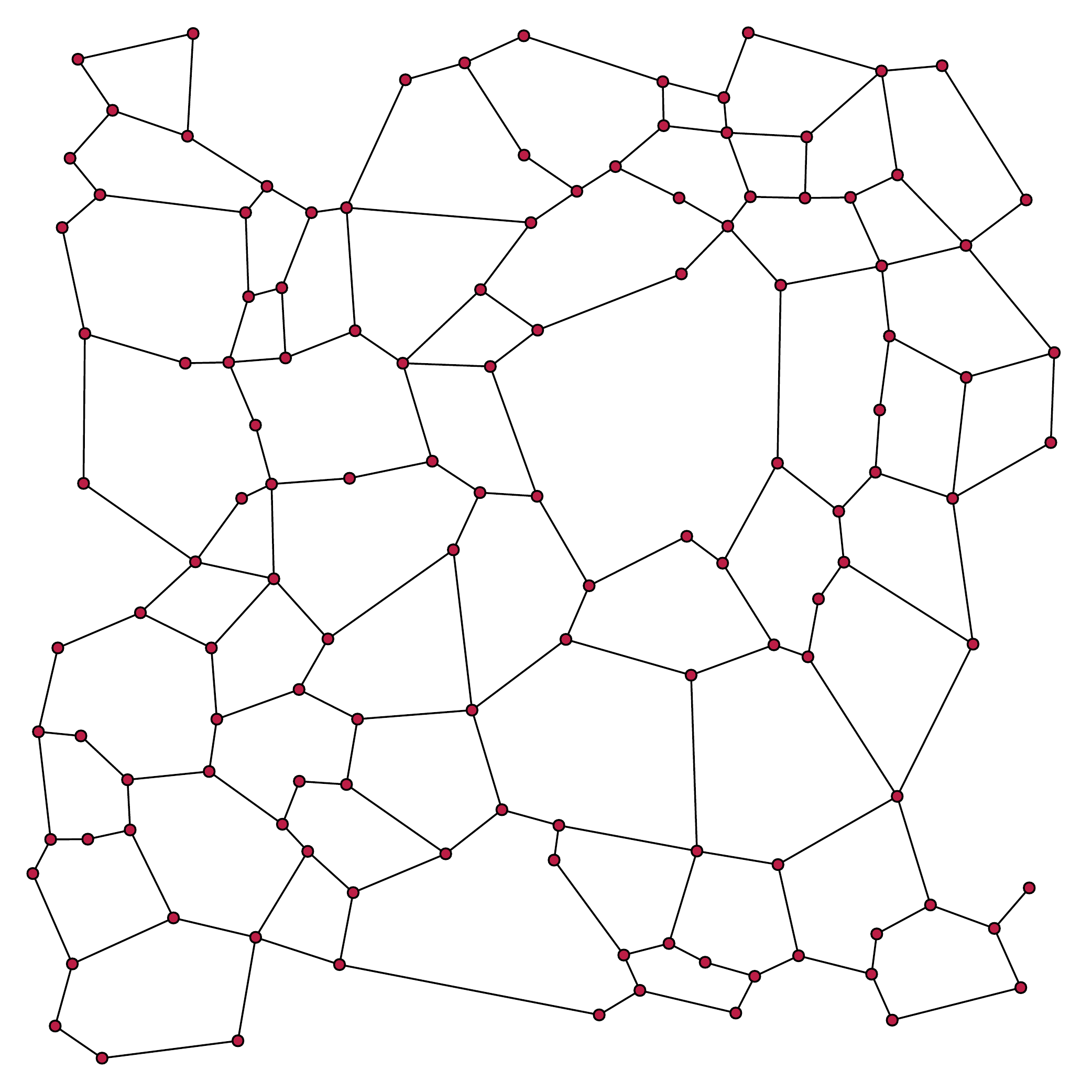}
\qquad
\includegraphics[width=0.45\textwidth]{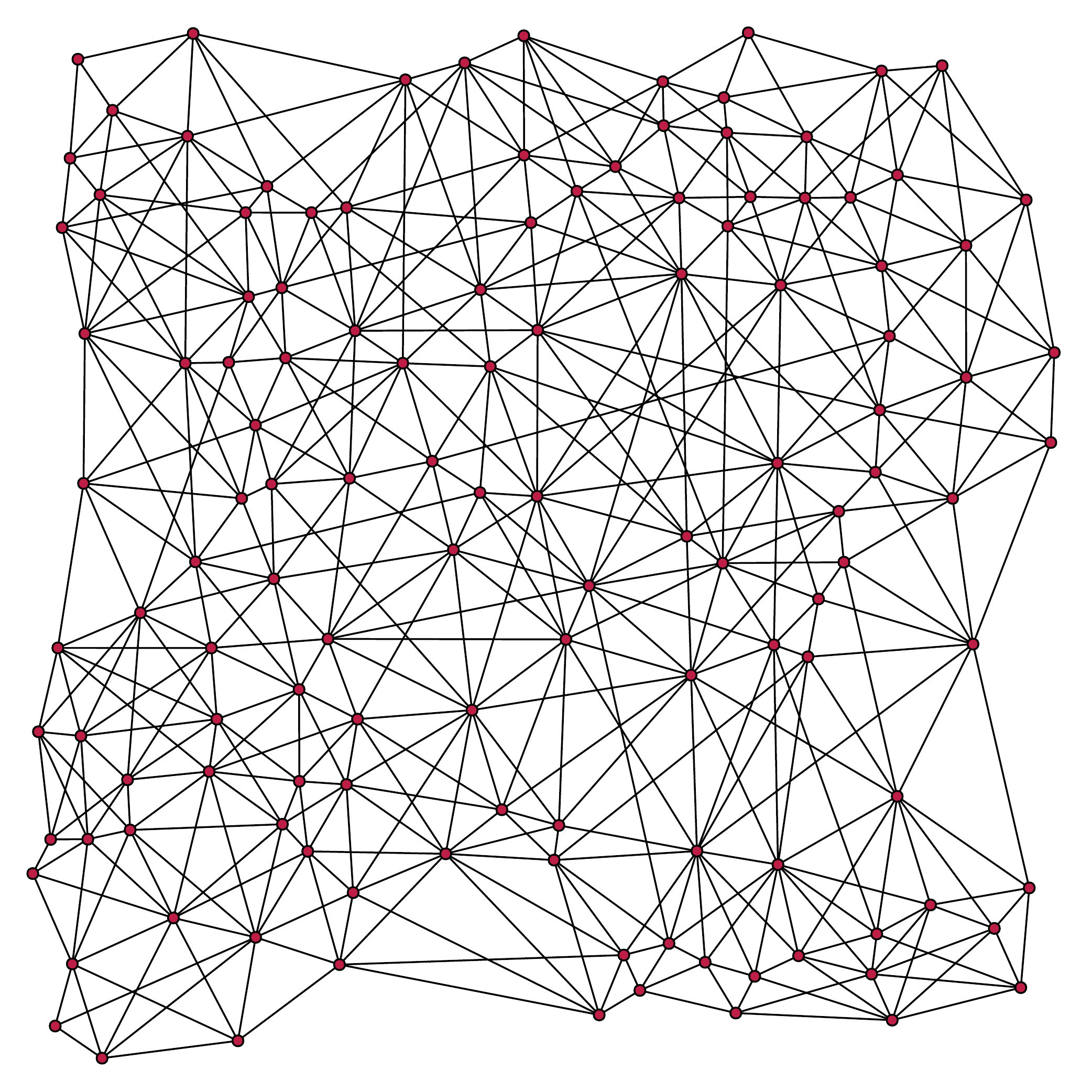}
\caption{Greedy spanners of 128 random points with stretch factor $2$ (left) and $1.1$ (right)}
\label{fig:greedy}
\end{figure}

\begin{figure}[t]
\centering
\includegraphics[width=0.45\textwidth]{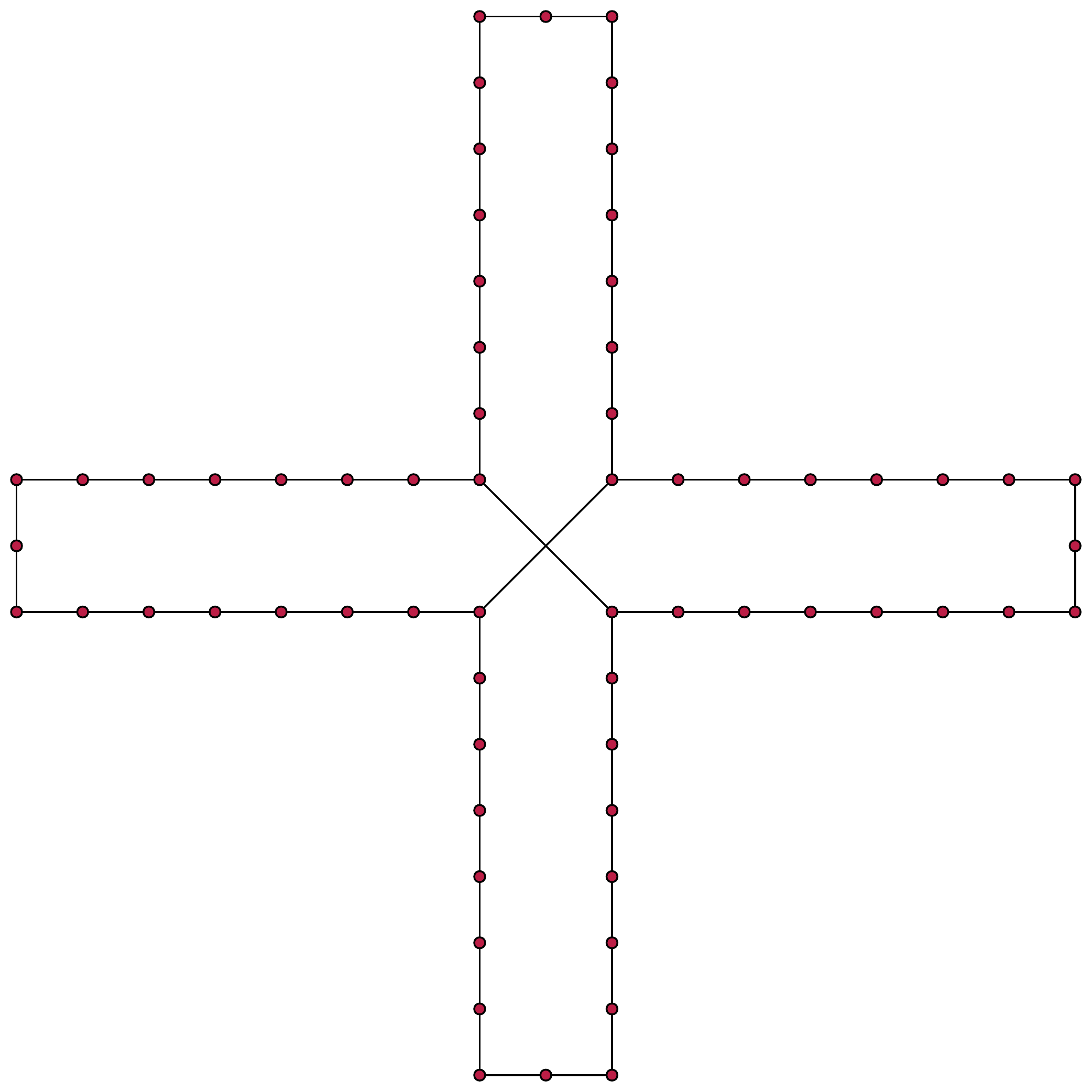}
\caption{Nonplanar greedy spanner with stretch factor 11.3}
\label{fig:cross}
\end{figure}

Spanners can be defined in any metric space, but they are often located in a geometric space, where a heavy or undesirable network is given and finding a sparse and light-weight spanner and working with it instead of the actual network makes the computation easier and faster. Finding light-weight geometric spanners has been a topic of interest in many areas of computer science, including communication network design and distributed computing. These subgraphs have few edges and are easy to construct,
leading them to appear in a wide range of applications since they were introduced \cite{chew1989planar,keil1988approximating,peleg1989graph}.  In wireless ad hoc networks $t$-spanners are used to design sparse networks with guaranteed connectivity and guaranteed bounds on routing length \cite{alzoubi2003geometric}. In distributed computing spanners provide communication-efficiency and time-efficiency through the sparsity and the bounded stretch property \cite{baswana2010additive,elkin20041e,awerbuch1998near,elkin2006efficient}. There has also been extensive use of geometric spanners in the analysis of road networks \cite{eppstein1999spanning,abam2009region,chechik2010fault}. In robotics, geometric spanners helped motion planners to design near-optimal plans on a sparse and light subgraph of the actual network \cite{dobson2014sparse,marble2013asymptotically,das1997visibility}. Spanners have many other applications including computing almost shortest paths \cite{elkin2005computing,cohen1998fast,roditty2004dynamic,feigenbaum2005graph}, %distance oracles \cite{thorup2005approximate,roditty2005deterministic,patrascu2010distance,kapralov2012spectral},
and overlay networks \cite{braynard2002opus,wang2005network,jia2003local}.

Researchers have developed various construction techniques for spanners, depending on the specific additional properties needed in these applications. Well-separated pair decomposition, $\theta$-graphs, and greedy spanners are among the most well-known of these geometric spanner constructions. Here, we focus on the greedy spanner. It was first introduced by Alth\"{o}fer \cite{althofer1990generating,althofer1993sparse} and Bern, generalizing a pruning strategy used by Das and Joseph \cite{das1989triangulations} on a triangulation of the planar graph \cite{eppstein1999spanning}.

A greedy spanner can be constructed by running the greedy spanner algorithm (\autoref{alg:greedy}) on a set of points on the Euclidean plane. This short procedure adds edges one at a time to the spanner it constructs, in ascending order by length. For each pair of vertices, in this order, it checks whether that pair already satisfies the bounded stretch inequality using the edges already added. If not, it adds a new edge connecting the pair. Therefore, by construction, each pair of vertices satisfies the inequality, either through previous edges or (if not) through the newly added edge. The resulting graph is therefore a $t$-spanner. Examples of the results of this algorithm, for two different stretch factors, are shown in \autoref{fig:greedy}. Although the 2-spanner in the figure is planar, this is not true for 2-spanners in general:
there exist point sets with non-planar greedy $t$-spanners for arbitrarily large values of~$t$ (\autoref{fig:cross}), and by placing widely-spaced copies of the same construction within a single point set, one can construct point sets whose greedy $t$-spanners have linearly many crossings, for arbitrarily large values of~$t$.

\begin{algorithm}
\caption{The naive greedy spanner algorithm.}\label{alg:greedy}
\begin{algorithmic}[1]
\Procedure{Naive-Greedy}{$V$}
\State Let $S$ be a graph with vertices $V$ and edges $E=\{\}$
\For {each pair $(P,Q)\in V^2$ in increasing order of $d(P,Q)$}
\If {$d_S(P,Q) > t\cdot d(P,Q)$}
\State Add edge $PQ$ to $E$%
\EndIf%
\EndFor%
\Return S%
\EndProcedure%
\end{algorithmic}
\end{algorithm}

A na{\"\i}ve implementation of the greedy spanner algorithm runs in time $\mathcal{O}(n^3\log n)$, where $n$ is the number of given points \cite{bose2010computing}. Bose et al. \cite{bose2010computing} improved the running time of \autoref{alg:greedy} to near-quadratic time using a bounded version of Dijkstra's algorithm. Narasimhan et al. proposed an approximate version of the greedy spanner algorithm that reached a running time of $\mathcal{O}(n\log n)$, based on the use of  approximate shortest path queries \cite{das1997fast,gudmundsson2002fast,narasimhan2007geometric}.

Despite the simplicity of \autoref{alg:greedy}, Farshi and Gudmundsson~\cite{farshi2005experimental} observed that in practice, greedy spanners are surprisingly good in terms of the number of edges, weight, maximum vertex degree, and also the number of edge crossings. Many of these properties have been proven rigorously. Filster and Solomon \cite{filtser2016greedy} proved that greedy spanners have size and lightness that is optimal to within a constant factor for worst-case instances. They also achieved a near-optimality result for greedy spanners in spaces of bounded doubling dimension. Borradaile, Le, and Wulff-Nilsen~\cite{borradaile2019greedy} recently proved optimality for doubling metrics, generalizing a result of Narasimhan and Smid~\cite{narasimhan2007geometric}, and resolving an open question posed by Gottlieb \cite{gottlieb2015light}, and Le and Solomon showed that no geometric $t$-spanner can do asymptotically better than the greedy spanner in terms of number of edges and lightness \cite{le2019truly}. However, past work has not proven rigorous bounds on the number of crossings of greedy spanners.

One reason for particular interest in bounds on the number of crossings is the close relation, for geometric graphs in the plane, between crossings and \emph{separators}. The well-known planar separator theorem of Lipton and Tarjan~\cite{lipton1979separator} states that any planar graph (that is, a geometric graph with no crossings) can be partitioned into subgraphs whose size is at most a constant fraction of the total by the removal of $O(\sqrt n)$ vertices. This property is central to the efficiency of many algorithms on planar graphs~\cite{goodrich1995planar,dvorak2016strongly,eppstein2010linear,eppstein2008studying,klein2013structured}, and applied as well in multiple computational geometry problems~\cite{frieze1992separator,arikati1996planar,kirkpatrick1983optimal}.
Analogous separator theorems have been extended from planar graphs to graphs with few crossings per edge~\cite{dujmovic2017structure}, or more generally to graphs with sparse patterns of crossings~\cite{eppstein2017crossing,bae2018gap}.
Past work has not shown that greedy spanners have small separators, but as we will show, bounds on their crossings can be used to show that they do.

\subsection{Our Contribution}
In this paper we prove that greedy $t$-spanners in the Euclidean plane have few crossings, for any $t>1$, and we use this result (together with a result of Eppstein and Gupta~\cite{eppstein2017crossing} on graphs with sparse patterns of crossings) to prove that greedy spanners in the Euclidean plane have small separators.
In particular, we prove:
\begin{itemize}
    \item \textbf{Claim 1.} Each edge in a greedy spanner can be crossed by only $O(1)$ edges of equal or greater length, where the constant in the $O(1)$ depends only on $t$, the stretch factor of the spanner. More precisely as $t\to 1$ there are $O(1/(t-1)^2)$ edges that cross the given edge and are longer than it by a factor of $\Omega(1/(t-1))$ (\autoref{thm:longer}), and $1/(t-1)^{O(1)}$ edges that cross the given edge and have length at least $\epsilon$ times it, for any constant $\epsilon>0$ (\autoref{thm:nearly-as-long}).
    \item \textbf{Claim 2.} For some choices of $t$, there exist greedy spanners in which some edges are crossed by a linear number of (significantly shorter) edges (\autoref{thm:many-crossings}).
   \item \textbf{Claim 3.} Every $n$-vertex greedy spanner, and every $n$-vertex subgraph of a greedy spanner, can be partitioned into connected components of size at most $cn$ for a constant $c<1$ by the removal of $O(\sqrt n)$ vertices. Again, the constant factor in the $O(\sqrt n)$ term depends only on the stretch factor of the spanner.  Moreover, a separator hierarchy for the greedy spanner can be constructed from its planarization in near-linear time (\autoref{th:sep}).
\end{itemize}

It is known that the spanners that are constructed by some other methods, i.e. semi-separated pair decomposition \cite{abam2012new} and hierarchical decomposition \cite{furer2007spanners}, have small $\mathcal{O}(\sqrt{n})$-separators in two dimensions. Although experimental results of Farshi and Gudmundsson on greedy spanners of random point sets had shown the number of crossings to be small in practice~\cite{farshi2005experimental} our results are the first theoretical results on this property, the first to study crossings for worst-case and not just random instances, and the first to prove that greedy spanners have small $\mathcal{O}(\sqrt{n})$-separators.

\subsection{Intuition}

Our proof that edges can be crossed by only a bounded number of edges of greater or equal length splits into two cases, one for crossings by edges of significantly greater length and another for crossings by edges of similar length.

For edges of significantly greater length, we divide the greedy spanner edges that might cross the given edge into a constant number of nearly-parallel sets of edges, and prove the bound separately within each such set. We show that, within a set of nearly-parallel long edges that all cross the given edge, the edges can be totally ordered by their projections onto a base line, because edges whose endpoints project to nested intervals would contradict the greedy property of the spanner (the inner of two nested edges could be used to shortcut the outer one). By similar reasoning, the endpoints of any two nearly-parallel long crossing edges are separated by a distance that is at least a constant fraction of the length of the smaller edge. This geometric growth in the separation of the endpoints leads to a system of inequalities on the lengths of the edges that can only be satisfied when the number of crossing edges is bounded by a constant.

For edges of comparable length to the crossed edge, we use a grid to partition the crossing edges into a constant number of subsets of edges, such that within each subset all edges have pairs of endpoints that are close to each other relative to the length of the edge, and we show that each of these subsets can contain only a unique edge.

Our construction showing that a single edge can be crossed linearly many times is based on the combination of three ``zig-zag'' sets of points, evenly spaced in their $x$-coordinates and alternating between two different $y$-coordinates. In the top and bottom zig-zag, the distance along the zigzag between two consecutive points with the same $y$-coordinates is exactly $t$ times the difference between their $x$-coordinates, while in the middle zig-zag it is slightly greater. The greedy spanner for this point set contains the zig-zag edges, plus a single long edge crossing all of the middle edges, for a pair of points that are far enough from each other along the middle zig-zag for their Euclidean distance to be almost the same as their difference in $x$-coordinates (differing by a number smaller than the amount by which a single edge of the middle zig-zag exceeds $t$ times its difference in $x$-coordinates).

The results on separators follow from previous results on the existence of separators in graphs whose edge intersection graphs have bounded degeneracy~\cite{eppstein2017crossing}.

\section{Preliminaries}
As we mentioned earlier, $t$-spanners can be defined in any metric space. For a given graph $G$, a $t$-spanner is defined in the following way,
\begin{definition}[$t$-spanner]
Given a metric graph $G=(V,E,d)$, i.e. weighted graph with distances as weights, a $t$-spanner is a spanning subgraph $S$ of $G$ such that for any pair of vertices $u, w\in V$,
\[d_G(u,w)\leq t\cdot d(u,w)\]
where $d_G(u,w)$ is the length of the shortest path in $G$ between $u$ and $w$.
\end{definition}

Then the greedy spanner on a given set of points $V$ can be defined in the following way,
\begin{definition}[greedy spanner]
Given a set of points $V$ in any metric space, a greedy spanner on $V$ is a $t$-spanner that is an output of \autoref{alg:greedy}.
\end{definition}

Here we restrict the problem to geometric graphs and we take advantage of inequalities that hold in geometric space.

We consider the natural embedding that the greedy spanner inherits from its vertices. Edges are drawn as straight segments between the two points corresponding to the two endpoints of the edge. We say two edges of the spanner cross or intersect if their corresponding segments intersect at some interior point. The crossing graph of a given embedding can be defined in this way,
\begin{definition}[crossing graph]
Given a graph $G(V,E)$ and its Euclidean embedding, the crossing graph $Cr(G)$ is a graph $G'(E,C)$ whose vertices are the edges of the original graph and for each two vertices $e,f\in E$ there is an edge between them if and only if they intersect with each other in the embedding given for $G$.
\label{def:crossgraph}
\end{definition}

Most of the proofs here use a lemma that we call the \emph{short-cutting lemma}, which is simple but very useful in greedy spanners. The lemma is proven in \cite{narasimhan2007geometric} and it states that a $t$-spanner edge cannot be shortcut by some other edges of the spanner by a factor of $t$. Formally,

\begin{lemma}[short-cutting]
An edge $AB$ of a greedy $t$-spanner cannot be shortcut by some other spanner edges by a factor of $t$, i.e. there is no constant $k$ and points $A=P_0,P_1,\dotsc,P_k=B$ that $P_0P_1, P_1P_2,\dotsc, P_{k-1}P_k$ are all spanner edges distinct from $AB$, and
\[ \sum_{i=0}^{k-1} |P_iP_{i+1}| \leq t\cdot |AB| \]
\label{lem:shortcut}
\end{lemma}

If some of the segments $P_iP_{i+1}$ are not included in the spanner, the same argument still works but a factor $t$ appears before the term $|P_iP_{i+1}|$ in the summation. So

\begin{corollary}[Extended short-cutting]
Given a greedy $t$-spanner $S$ and an edge $AB$ of $S$, there cannot be a constant $k$ and points $A=P_0,P_1,\dotsc,P_k=B$ such that
\[\sum_{P_iP_{i+1}\in S} |P_iP_{i+1}| + t\cdot\sum_{P_iP_{i+1}\notin S} |P_iP_{i+1}| \leq t\cdot |AB|\]
\label{cor:shortcut}
\end{corollary}

The proof of Lemma \ref{lem:shortcut} and Corollary \ref{cor:shortcut} are included in Appendix \ref{sec:lemmas} for reference. In the following section we consider intersections between an arbitrary edge of a greedy spanner and sufficiently larger edges, and we show a constant bound on the number of intersections per edge. In \autoref{sec:same} we again prove a constant bound for the number of intersections between a spanner edge and other edges of almost the same length. Finally, in \autoref{sec:smaller} we introduce an example in which the number of intersections with smaller edges can be more than any constant bound, completing our analysis. In \autoref{sec:app} we introduce some new results and improvements based on the constant bound we provided earlier.

\section{Few intersections with long edges}
\label{sec:longer}

In this section, we prove an upper bound on the number of intersections of an edge with sufficiently larger edges. We will specifically show that the number of intersections, in this case, has a constant bound that only depends on $t$. Later in \autoref{sec:same} we prove a constant bound also exists for the intersections with the edges that have almost the same length of the intersecting edge. Hence we prove our first claim.

In this setting, we consider an arbitrary edge $AB$ of the spanner, and we are interested in counting the number of intersections that $AB$ may have with sufficiently larger edges, i.e. edges $PQ$ that intersect $AB$ at some interior point with $|PQ| > c\cdot |AB|$ for some constant $c>1$ which we will specify later.

First, we only consider a set of \emph{almost-parallel} spanner segments that cross $AB$, where we define the term \emph{almost-parallel} below, and we put a bound on the number of these segments. Then we generalize the bound to hold for all large spanner segments that cross $AB$.

\subsection{Definitions}

\begin{definition}[almost-parallel]
\label{def:almostpar}
We say a pair of arbitrary segments $PQ$ and $RS$ in the plane are \emph{almost-parallel} or \emph{$\theta$-parallel} if there is an angle of at most $\theta$ between them. We say a set of segments are \emph{almost-parallel} if every pair of segments chosen from the set are almost-parallel.
\end{definition}

For any set of almost-parallel segments, we define a baseline to measure the angles and distances with respect to that line.

\begin{definition}[baseline]
\label{def:baseline}
Given a set of almost-parallel (or $\theta$-parallel) segments in the plane, denoted by $S$, the \emph{baseline} $b(S)$ of the set of segments $S$ is the segment with the smallest slope.
\end{definition}

We use the uniqueness of the segment chosen in Definition \ref{def:baseline} and we emphasize that any other definition works if it determines a unique segment for any almost-parallel set of segments.

In \autoref{sec:order}, we define a total ordering on a set of almost-parallel segments that cross a spanner segment $AB$. Once we have sorted these segments based on the ordering, in \autoref{sec:lowerbd} we prove the distance between the endpoints of two consecutive segments is at least a constant fraction of the length of the smaller segment. Putting together these two parts, in \autoref{sec:together} we prove there cannot be more than a constant number of segments in the sequence.

\subsection{A total ordering on almost-parallel intersecting segments}
\label{sec:order}
In this section, we define an ordering on a set of almost-parallel segments of the $t$-spanner. The ordering is based on the order of the projections of the endpoints of the segments on the baseline corresponding to the segments. We first define the ordering and then we use Lemma \ref{lem:interorder} and Lemma \ref{lem:order} to prove that it is a total ordering when the set of almost-parallel segments are all crossing a given segment of the spanner.

Consider a set of almost-parallel spanner segments that cross some spanner segment. One can define an ordering on this set of almost-parallel segments, which we call the \emph{endpoint-ordering}, based on how their endpoints are ordered along the direction they are aligned to. We formulate the definition in the following way,

\begin{definition}[endpoint-ordering]
\label{def:ordering}
Let $S=\{P_iQ_i:i=1,2,\dotsc,k\}$ be a set of almost-parallel segments. Also let $l$ be the baseline of $S$, $b(S)$. Define the \emph{endpoint-ordering} $\mathcal{R}$ between two segments $P_iQ_i$ and $P_jQ_j$ by projecting the endpoints $P_i,P_j,Q_i,Q_j$ to the baseline $l$ and comparing the order of the projected points $P_i',P_j',Q_i',Q_j'$ along an arbitrary direction of the baseline $l$,
\begin{itemize}
    \item $P_iQ_i<_{\mathcal{R}}P_jQ_j$ if the projections are ordered as $P_i'P_j'Q_i'Q_j'$ or $P_i'Q_i'P_j'Q_j'$.
    \item $P_iQ_i>_RP_jQ_j$ if they are ordered as $P_j'P_i'Q_j'Q_i'$ or $P_j'Q_j'P_i'Q_i'$. (\autoref{fig:ordering})
\end{itemize}
\begin{figure}[t]
\centering\includegraphics[width=0.5\textwidth,trim=2.3cm 4.5cm 3cm 3cm,clip]{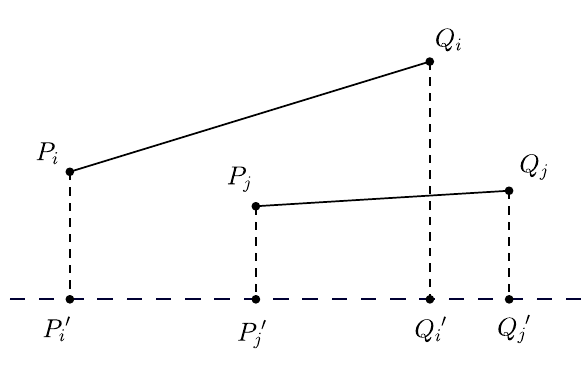}
\caption{Ordering segments by projecting on the baseline $l$, here $P_iQ_i<_{\mathcal{R}}P_jQ_j$.}
\label{fig:ordering}
\end{figure}
\end{definition}

We claim that the endpoint-ordering is a total ordering on the set of almost-parallel segments. This basically means that after projecting two almost-parallel segments on the baseline, none of the resulting projections would lie completely inside the other one. Other cases correspond to a valid endpoint-ordering.

In order to prove this, first, we prove a simpler case when the two segments intersect with each other. This assumption will help to significantly simplify the proof. Later we use this lemma to show the original claim is also true.

\begin{lemma}
Let $MN$ and $PQ$ be two intersecting segments from a set of $\theta$-parallel spanner segments. Also assume that $\theta < \frac{t-1}{2t}$ where $t$ is the stretch factor of the spanner. Then $MN$ and $PQ$ are endpoint-ordered, i.e. the projection of one of the segments on the baseline of the set cannot be included in the projection of the other one.
\label{lem:interorder}
\end{lemma}
\begin{proof}

We prove the lemma by contradiction. Without loss of generality suppose that the projections of $P$ and $Q$ on some baseline $l$ are both between the projections of $M$ and $N$ (on the same baseline). We show that $MN$ can be shortcut by $PQ$ by a factor of $t$, i.e.
\[t\cdot |MP| + |PQ| + t\cdot |QN| \leq t\cdot |MN|\]
\begin{figure}[t]
\centering
\includegraphics[width=0.6\textwidth,trim=2cm 5cm 2.5cm 3cm,clip]{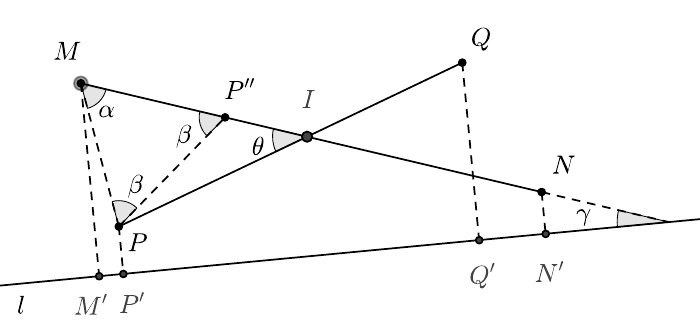}
\caption{Proof of Lemma \ref{lem:interorder}.}
\label{fig:interorder}
\end{figure}

Let $P', Q', M'$, and $N'$ be the corresponding projections of $P,Q,M$, and $N$ on $l$, respectively (\autoref{fig:interorder}). Also let $I$ be the intersection point and $\alpha=\angle PMI$, and also $\gamma$ to be the angle between $MN$ and the baseline, according to the figure. By the assumption $P'$ is between $M'$ and $N'$, so $\alpha \leq \pi/2+\gamma \leq \pi/2+\theta$. Let also $P''$ be the point on $MN$ s.t. $|MP''| = |MP|$ and $\beta=\angle MPP''=\angle MP''P$. Then by sine law,
\begin{align}
\begin{split}
\frac{|MI| - |MP|}{|PI|} = \frac{|P''I|}{|PI|} &= \frac{\sin (\beta-\theta)}{\sin \beta} = \frac{\sin (\pi/2-\alpha/2-\theta)}{\sin (\pi/2-\alpha/2)} = \frac{\cos (\alpha/2 + \theta)}{\cos (\alpha/2)} \\
&= \cos\theta - \sin\theta\tan(\alpha/2)
\end{split}
\label{shortcut-intersected-e3}
\end{align}
but we have,
\begin{equation}
\cos\theta \geq 1-\theta^2/2 \geq 1-\theta/4
\label{shortcut-intersected-e4}
\end{equation}
as $\theta<\frac{t-1}{2t}<1/2$. Also,
\begin{equation}
\tan(\alpha/2)\leq \tan(\pi/4+\theta/2) = \tan(\pi/4+1/4) < \frac{7}{4}
\label{shortcut-intersected-e5}
\end{equation}
Putting together \autoref{shortcut-intersected-e3}, \autoref{shortcut-intersected-e4}, and \autoref{shortcut-intersected-e5}, also using $\sin\theta \leq \theta$,
\[\frac{|MI| - |MP|}{|PI|} \geq (1-\theta/4) - (\frac{7}{4})\theta = 1 - 2\theta > \frac{1}{t}\]
which is equivalent to $t\cdot |MI| - t\cdot |MP| \geq |PI|$. Similarly, $t\cdot |NI| - t\cdot |NQ| \geq |QI|$. Adding together,
\[t\cdot |MN| - t\cdot |MP| - t\cdot |NQ| \geq |PQ|\]
which is what we are looking for.
\end{proof}

Lemma \ref{lem:interorder} assumes that segments intersect at some interior point. In order to prove the totality of the ordering, we also need to prove the claim when the segments do not intersect with each other. Instead, in this case, both segments intersect some spanner edge. We use Lemma \ref{lem:interorder} to prove this in the Lemma \ref{lem:order}.

\begin{lemma}
Let $MN$ and $PQ$ be two segments chosen from a set of $\theta$-parallel spanner segments that cross a spanner edge $AB$. Also assume that $\theta<\frac{t-1}{2(t+1)}$, and $\min(|MN|, |PQ|)\geq\frac{3t(t+1)}{t-1}|AB|$, where $t$ is the spanner parameter. Then $MN$ and $PQ$ are endpoint-ordered.
\label{lem:order}
\end{lemma}

The proof of this lemma is included in Appendix \ref{sec:lemmas}. Based on Lemma \ref{lem:order} it is easy to prove the main result of this section, Proposition \ref{prop:totality}.

\begin{proposition}
Given an arbitrary edge $AB$ of a $t$-spanner, for a set of sufficiently large almost-parallel spanner edges that intersect $AB$, the \emph{endpoint-ordering} we defined in Definition \ref{def:ordering} is a total ordering.
\label{prop:totality}
\end{proposition}

\begin{proof}

Totality requires reflexivity, anti-symmetry, transitivity, and comparability. Reflexivity and transitivity are trivial because of the projection. Anti-symmetry and comparability follow directly from Lemma \ref{lem:order}.
\end{proof}

Now that we have ordered the set of almost-parallel spanner segments, we can prove a lower bound on the distance of two ordered segments. Later we prove a bound on the number of these segments based on the resulting distance lower bound.

\subsection{Lower bounding the distance of endpoints of two crossing segments}
\label{sec:lowerbd}

In \autoref{sec:order} we restricted the problem to a set of almost-parallel spanner segments that intersect another spanner segment, and we defined an ordering on these segments. The next step is to find a lower bound on the distance of two almost-parallel segments that intersect some spanner segment $AB$. The idea is to show that both endpoints of two ordered segments cannot be arbitrarily close, and hence there cannot be more than a constant number of them in a sequence.

More specifically, we show in Proposition \ref{prop:lowerbd} that the corresponding endpoints of two almost-parallel spanner segments that both cross the same spanner segment should have a distance of at least a constant fraction of the length of the smaller segment, otherwise the longer segment could be shortcut by the smaller one, which is indeed a contradiction. A weaker version of this lemma is proven in \cite{narasimhan2007geometric} and it is called ``gap property'', but the inequality we show here is actually stronger.

First we propose a geometric inequality in Lemma \ref{lem:lowerbd} that helps to prove the proposition. Then we complete the proof of the proposition at the end of this section.

\begin{lemma}
Let $MN$ and $PQ$ be two segments in the plane with angle $\theta$. Then
\[\left||MN| - |PQ|\right| > \left||MP| - |NQ|\right| - 2\sin(\theta/2)\cdot |PQ|\]
\label{lem:lowerbd}
\end{lemma}
\begin{proof}
By swapping $MN$ and $PQ$, it turns out that the case where $|MN|\geq |PQ|$ is stronger than $|MN|\leq |PQ|$. So without loss of generality, let $|MN|\geq |PQ|$ and by symmetry $|MP|\geq |NQ|$. Let $Q'$ be the rotation of $Q$ around $P$ by $\theta$, so that $PQ'$ and $MN$ are parallel, and $|PQ'|=|PQ|$ (\autoref{fig:lowerbd}).
\begin{figure}[t]
\centering
\includegraphics[width=0.6\textwidth,trim=2.7cm 5.7cm 3cm 2.6cm,clip]{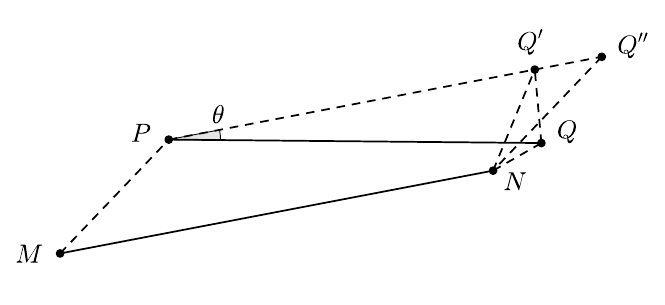}
\caption{Proof of Lemma \ref{lem:lowerbd}.}
\label{fig:lowerbd}
\end{figure}
Let $Q''$ be the point on the ray $PQ'$ where $|PQ''| = |MN|$. As a result $Q''$ and $P$ will be on different sides of $Q'$. By the triangle inequality,
\begin{align*}
|MN| - |PQ| &= |PQ''| - |PQ'| = |Q'Q''| \geq |NQ''| - |NQ'| \\
&= |MP| - |NQ'| \geq |MP| - (|NQ| + |QQ'|) \\
&= |MP| - |NQ| - 2|PQ|\cdot\sin(\theta/2) %\\
%&\geq |MP| - |NQ| - \theta\cdot |PQ|
\end{align*}
\end{proof}

Now we state and prove Proposition \ref{prop:lowerbd}. As we mentioned earlier, the idea is to show one of the segments can be shortcut by the other one if one of the matching endpoints is very close. In the simplest case when the segments are two opposite sides of a rectangle, it is easy to see that a distance of $\frac{t-1}{2}|PQ|$ on both sides is required to prevent short-cutting. In the general case, when the segments are placed arbitrarily, Proposition \ref{prop:lowerbd} holds.

\begin{proposition}
Let $MN$ and $PQ$ be two $\theta$-parallel spanner segments. The matching endpoints of these two segments cannot be closer than a constant fraction of the length of the smaller segment. More specifically,
\[\min(|MP|, |NQ|) \geq \frac{t-1-2\sin(\theta/2)}{2t}\min(|MN|,|PQ|)\]
\label{prop:lowerbd}
\end{proposition}

\begin{proof}
Without loss of generality and by symmetry, let $|NQ|\leq |MP|$. Suppose, on the contrary, that $|NQ| < \frac{t-1-2\sin(\theta/2)}{2t}|PQ|$. Then,
\begin{align*}
t\cdot |MP| + |PQ| + t\cdot |NQ| &\leq t\cdot(|MN|-|PQ|+|NQ|+2\sin(\theta/2)\cdot |PQ|)+|PQ|+t\cdot |NQ| \\
&= t\cdot |MN| - (t-1-2\sin(\theta/2))|PQ| + 2t\cdot |NQ| \\
&< t\cdot |MN|
\end{align*}
So $MN$ can be shortcut by $PQ$ within a factor of $t$ which contradicts the extended short-cutting lemma for the edge $MN$ and the path $MPQN$.
\end{proof}

So far, in Proposition \ref{prop:lowerbd} we proposed an ordering on the set of almost-parallel spanner segments that cross a given edge and we proved each of these segments has a significant distance from the other ones. In the next section we put together these results and we find a constant upper bound on the number of these segments.

\subsection{Putting together}
\label{sec:together}

Based on the ordering proposed in \autoref{sec:order}, and the lower bound we proved in \autoref{sec:lowerbd}, we can show that the following constant upper bound on the number of intersections with sufficiently large edges holds.

If we look at one of the endpoints of the endpoint-ordered sequence of almost-parallel spanner segments, and we project them on the baseline, the distance of every two consecutive projected points cannot be smaller than a constant fraction of the length of the smaller segment, i.e. $|P_i'P_{i+1}'| \geq C\cdot\min(|P_iQ_i|, |P_{i+1}Q_{i+1}|)$ for all values of $i=0,1,\dotsc, k-1$. Summing up these inequalities leads to a bound on $k$, the number of segments.

\begin{theorem}
\label{thm:longer}
For sufficiently small $\theta$,
the number of sufficiently large $\theta$-parallel segments that intersect a given edge $AB$ of a $t$-spanner is limited by
\[\frac{4t}{(t-1-2\sin(\theta/2))\cos\theta}+1\]
By sufficiently large we specifically mean larger than $\frac{3t(t+1)}{t-1}|AB|$.
\label{th:larger}
\end{theorem}
\begin{proof}
Let $P_iQ_i$s be the segments larger than $AB$ that intersect $AB$ at some angle in $[\alpha, \alpha + \theta)$. Let $P_0Q_0$ be the shortest edge among $P_iQ_i$s. Because of the total ordering, at least half of the segments are larger than $P_0Q_0$ with respect to the ordering $\mathcal{R}$, or at least half of them are smaller than $P_0Q_0$ with respect to $\mathcal{R}$. Without loss of generality, assume that half of the segments are larger than $P_0Q_0$ with respect to $\mathcal{R}$, and they are indexed by $i=1,2,\dotsc,(k-1)/2$. Also let $P'_i$s and $Q'_i$s be the projections of $P_i$s and $Q_i$s on the base line $l$. By Proposition \ref{prop:lowerbd}, for all $i$, $P_{i+1}$ is farther than $P_i$ by a constant fraction of $\min(|P_iQ_i|, |P_{i+1}Q_{i+1}|)$, so
\begin{align*}
\sum_{i=0}^{(k-3)/2} |P_{i}P_{i+1}| &> \frac{t-1-2\sin(\theta/2)}{2t}\sum_{i=0}^{(k-3)/2} \min(|P_iQ_i|, |P_{i+1}Q_{i+1}|) \\
&\geq \frac{t-1-2\sin(\theta/2)}{2t}\cdot\frac{k-1}{2}|P_0Q_0|
\end{align*}
If $k\geq\frac{4t}{(t-1-2\sin(\theta/2))\cos\theta}+1$,
\[\sum_{i=0}^{(k-3)/2} |P_{i}P_{i+1}| > \frac{1}{\cos\theta}|P_0Q_0|\]
or equivalently
\[|P_0Q_0| < \sum_{i=0}^{(k-3)/2} |P_{i}P_{i+1}|\cos\theta \leq \sum_{i=0}^{(k-3)/2} |P'_{i}P'_{i+1}| = |P'_0P'_{\frac{k-1}{2}}|\]
which is not possible, because $P'_0P'_{\frac{k-1}{2}}$ lies inside $P'_0Q'_0$ and so $|P'_0P'_{\frac{k-1}{2}}| \leq |P'_0Q'_0| \leq |P_0Q_0|$ which contradicts the last inequality above.
\end{proof}

The constraints on $\theta$ imposed by our earlier lemmas imply that, as $t\to 1$, we should choose $\theta$ proportional to $t-1$. Asymptotically, as $t\to 1$, the number of large segments of all angles that intersect $AB$ is $O(1/(t-1)^2)$, with one factor of $1/(t-1)$ coming from the bound in the theorem and the second factor coming from the number of different classes of nearly-parallel segments.

\subsection{Almost-equal length edges}
\label{sec:same}

In the previous subsections, we proved a bound on the number of intersections with relatively larger edges. Here we prove a constant bound on the number of intersections with edges that are nearly the same length as the length of the intersecting edge. Later in \autoref{sec:smaller} we consider intersections with relatively smaller edges, which completes our analysis for this problem.

For same-length intersections Lemma \ref{lem:order} does not hold anymore, hence the endpoint-ordering is not necessarily a total ordering in this case. Since totality is a key requirement for the rest of the proof the same proof will not work anymore. But Proposition \ref{prop:lowerbd} still holds as it has no assumption on the ordering of the segments.

Our idea is to partition the neighborhood of $AB$ into a square network, such that no two spanner segments can have both endpoints in the same squares (\autoref{fig:partition}). If this happens, then by Proposition \ref{prop:lowerbd} one of the segments should be shortcut by the other one, leading to a contradiction because both segments are already included in the spanner.

We first prove a simpler version of Proposition \ref{prop:lowerbd} that does not include $\theta$ in the inequality, as we are not using the almost-parallel assumption and the value of $\theta$ can be large enough to make the inequality in Proposition \ref{prop:lowerbd} trivial. We will use this modified version to prove our claim.
\begin{lemma}
Given a greedy spanner with parameter $t$ and two spanner segments $MN$ and $PQ$,
\[\max(|MP|, |NQ|) \geq \frac{t-1}{2t}\min(|MN|, |PQ|)\]
\label{lem:subs-lower}
\end{lemma}
\begin{proof}
Suppose on the contrary that
\[\max(|MP|, |NQ|) < \frac{t-1}{2t}\min(|MN|, |PQ|)\]
Also, without loss of generality assume that $|MN|\geq |PQ|$. Then,
\[t\cdot |MP| + |PQ| + t\cdot |NQ| \leq (t-1)\min(|MN|,|PQ|) + |PQ| = t\cdot |PQ| \leq t\cdot |MN|\]
which contradicts the extended short-cutting lemma for the edge $MN$ and the path $MPQN$.
\end{proof}

\begin{proposition}
The number of spanner segments $PQ$ that cross a segment $AB$ of a $t$-spanner and that have length within $\alpha\cdot |AB|\leq |PQ| \leq \beta\cdot |AB|$ is limited by
\[\left[\frac{2\beta(2\beta+1)}{\alpha^2}\cdot\frac{8t^2}{(t-1)^2}\right]^2\]
where $t$ is the spanner parameter.
\label{prop:same-bound}
\end{proposition}
\begin{proof}
\begin{figure}[t]
\centering
\includegraphics[width=0.3\textwidth,trim=2.5cm 6cm 2.7cm 3.5cm,clip]{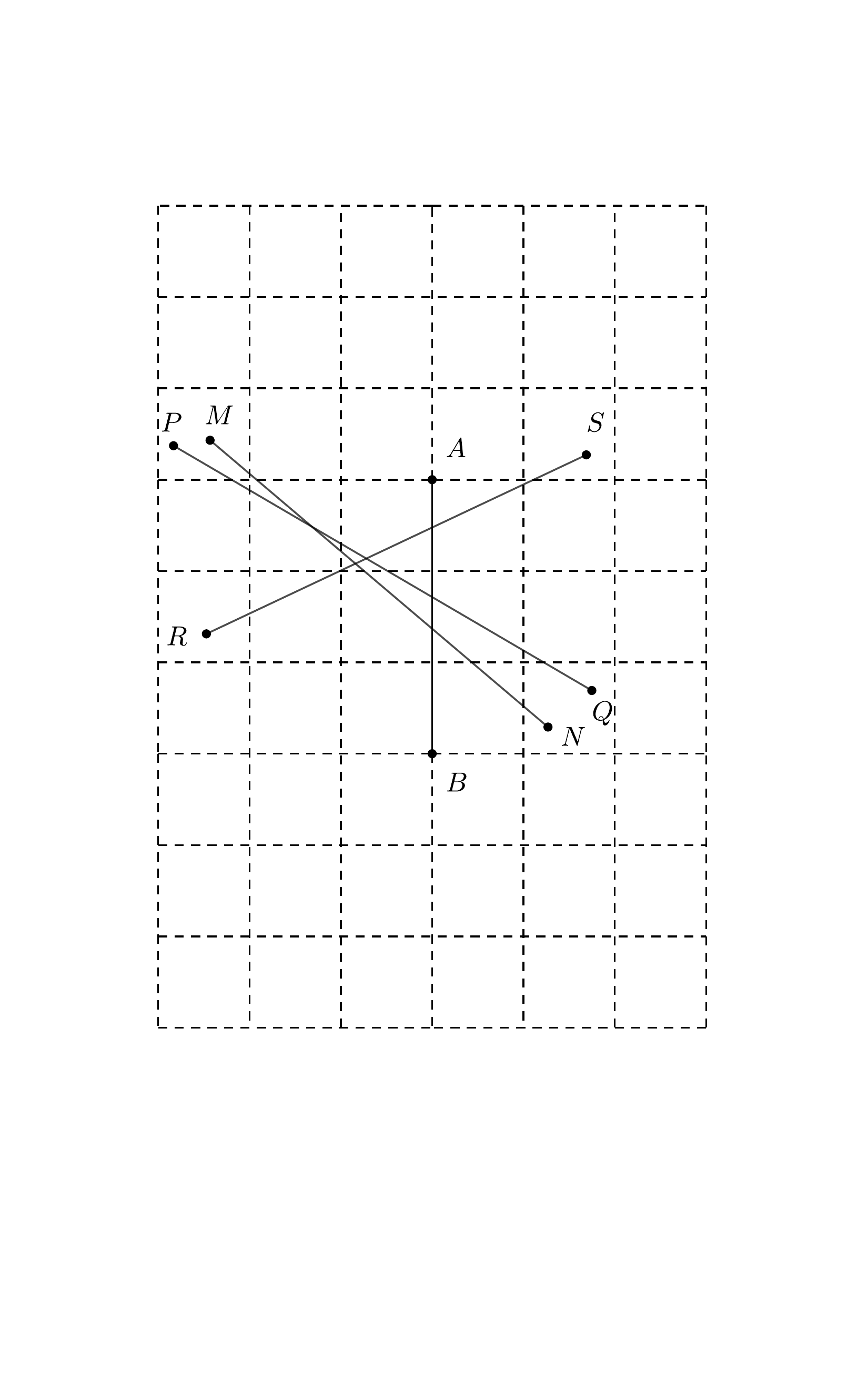}
\caption{Partition of the area around $AB$}
\label{fig:partition}
\end{figure}
Partition the area around $AB$ with squares of edge length $\frac{t-1}{2\sqrt{2}t}\cdot\alpha |AB|$ with edges parallel or perpendicular to $AB$. The area that an endpoint of a crossing segment can lie in is a rectangle of size $(2\beta +1)|AB|$ by $2\beta |AB|$ (\autoref{fig:partition}). The total number of squares in this area would be
\[\frac{2\beta(2\beta+1)}{\alpha^2}\cdot\frac{8t^2}{(t-1)^2}\]
But for each crossing segment the pair of squares that contain the two endpoints of the segment is unique. Otherwise two segments, e.g. $MN$ and $PQ$, will have both endpoints at the same pair, which means
\[\max(|MP|, |NQ|) < (\sqrt{2})(\frac{t-1}{2\sqrt{2}t}\cdot\alpha |AB|) = \frac{t-1}{2t}\cdot\alpha |AB| \leq \frac{t-1}{2t}\min(|MN|, |PQ|)\]
which cannot happen due to Lemma \ref{lem:subs-lower}. So the total number of pairs, and hence the total number of crossing segments, would be
\[\left[\frac{2\beta(2\beta+1)}{\alpha^2}\cdot\frac{8t^2}{(t-1)^2}\right]^2\]
\end{proof}

In Proposition \ref{prop:same-bound} both $\alpha$ and $\beta$ can be chosen arbitrarily, and the bound is a strictly increasing function of $\beta$ and a strictly decreasing function of $\alpha$. The bound tends to infinity when $\beta$ is large enough, and also when $\alpha$ is small enough. So it basically does not prove any constant bound for the cases that edges are very small or very large. But for the edges of almost the same length, it gives a constant upper bound.

Putting together the main results of \autoref{sec:longer} and \autoref{sec:same} we can prove the following bound for the number of intersections with not-relatively-small spanner segments.

\begin{theorem}
\label{thm:nearly-as-long}
Given a spanner segment $AB$ in the Euclidean plane and a positive constant $\epsilon$, the number of edges of length at least $\epsilon\cdot |AB|$ of the spanner that intersect $AB$ is $\mathcal{O}(\frac{t^{12}}{\epsilon^4 (t-1)^8})$.
\label{th:not-smaller}
\end{theorem}
\begin{proof}
By \autoref{th:larger} the number of intersections with edges $PQ$ such that $|PQ|\geq \frac{3t(t+1)}{t-1}|AB|$ is bounded by
\[C_1=\frac{4t}{(t-1-2\sin(\theta/2))\cos\theta}+1\in\mathcal{O}(\frac{t}{t-1})\]
On the other side, putting $\alpha=\epsilon$ and $\beta=\frac{3t(t+1)}{t-1}$ into Proposition \ref{prop:same-bound} implies that the number of intersections with edges larger than $AB$ and smaller than $\frac{3t(t+1)}{t-1}|AB|$ is at most
\[C_2=\left[\frac{2}{\epsilon^2}\left(\frac{3t(t+1)}{t-1}\right)\left(2\frac{3t(t+1)}{t-1}+1\right)\left(\frac{8t^2}{(t-1)^2}\right)\right]^2 \in \mathcal{O}(\frac{t^{12}}{\epsilon^4 (t-1)^8})\]
Hence the number of intersections with edges larger than $AB$ is at most $C_1+C_2$, which is $\mathcal{O}(\frac{t^{12}}{\epsilon^4 (t-1)^8})$.
\end{proof}

In \autoref{sec:longer} we proved the number of intersections with sufficiently large edges is bounded by a constant and now we completed the proof for all larger edges. In Appendix \ref{sec:smaller}, we show that the same argument does not work for intersections with arbitrarily smaller edges, and we provide an example that shows there can be an arbitrarily large number of intersections with smaller edges. This completes our analysis of the problem. In the following section, we show some of the applications of this result, most importantly, the sparsity of the crossing graph of the greedy spanner.

\section{Separators}
\label{sec:app}
In this section, we use the crossing bound that we proved in \autoref{th:not-smaller} to show that greedy spanners have small separators. First, we start with the definition of \emph{degeneracy}, which is a measure of sparsity of a graph.

\begin{definition}[degeneracy]
A graph $G$ is called $k$-degenerate, if each subgraph of $G$ has a vertex of degree at most $k$. The smallest value of $k$ for which a graph is $k$-degenerate is called the \emph{degeneracy} of the graph.
\end{definition}

The first important consequence of \autoref{th:not-smaller} is the constant degeneracy of the crossing graph of the greedy spanner, implying its sparsity and linearity of the number of edges, i.e. crossing.

\begin{theorem}
The crossing graph of a greedy $t$-spanner has a constant degeneracy.
\label{th:deg}
\end{theorem}
\begin{proof}
In any subgraph of the crossing graph, by \autoref{th:not-smaller} the node corresponding to the smallest edge has at most a constant number of neighbours.
\end{proof}

This, together with the result of \cite{eppstein2017crossing} implies the existence of sublinear separators for greedy spanners. A separator is a subset of vertices whose removal splits the graph into smaller pieces. A sublinear separator is a sublinear number of vertices with the same property. The splitting can be recursively performed on the smaller parts and a separator hierarchy can be constructed in this way, which effectively helps in the design of new recursive algorithms. The planarization of a graph, which is obtained by adding new vertices on the edge intersections of the graph, would help us to find such hierarchy in linear time, otherwise, a near-linear time algorithm would be used.

\begin{theorem}
Greedy spanners have separators of size $\mathcal{O}(\sqrt{n})$. Also, a separator hierarchy for them can be constructed from their planarization in linear time.
\label{th:sep}
\end{theorem}
\begin{proof}
By \autoref{th:deg} the crossing graph of the greedy $t$-spanner has a constant degeneracy, so by Theorem 6.9 of \cite{eppstein2017crossing} they have separators of size $\mathcal{O}(\sqrt{n})$. Also by the same theorem, a separator hierarchy for them can be constructed from their planarization in linear time.
\end{proof}

One of the basic algorithms that can be improved using the separator hierarchy is Dijkstra's single-source shortest path algorithm, which runs in $\mathcal{O}(n\log n)$ time on a graph with $n$ vertices. As a result of \autoref{th:sep} linear algorithms exist for finding single-source shortest path on greedy spanners, If the planarization has not already been found, it can be constructed in time $O(n\log^{(i)} n)$ for any constant $i$, where $\log^{(i)}$ denotes the $i$-times iterated logarithm, e.g. $\log^{(3)} n=\log\log\log n$~\cite{eppstein2010linear}.

\begin{corollary}
Single source shortest paths can be computed in time $O(n\log^{(i)} n)$ on a greedy spanner.
\label{cor:shortestpath}
\end{corollary}
\begin{proof}
This follows from the planarization algorithm and from the existence and construction of separators from planarizations by Corollary 6.10 of \cite{eppstein2017crossing}.
\end{proof}

\section{Conclusions}

We have shown that greedy $t$-spanners in the plane have linearly many crossings, and that the intersection graphs of their edges have bounded degeneracy but can have unbounded (and even linear) degree. As a consequence, we proved that these graphs have small separators.

Given these results, it is natural to ask whether higher-dimensional Euclidean greedy $t$-spanners also have small separators. This cannot be achieved through bounds on crossings, because in dimensions greater than two, graphs whose vertices are in general position can have no crossings. We leave this question open for future research.

\bibliographystyle{plainurl}
\bibliography{reference}

\appendix
\newpage

\section{Proof of some lemmas}
\label{sec:lemmas}

In this section we prove Lemma \ref{lem:shortcut} (short-cutting lemma) and Corollary \ref{cor:shortcut} (extended short-cutting) which are known but useful results that we used multiple times in this paper. We also include the proof of Lemma \ref{lem:order} which uses some geometric arguments similar to the proof of Lemma \ref{lem:interorder}.

\begin{proof}[Proof of Lemma \ref{lem:shortcut}]
Suppose on the contrary that such points exist. If $AB$ is larger than all other segments $P_iP_{i+1}$, then it should be added the last by the greedy algorithm, so when $AB$ is being added all $P_iP_{i+1}$s are already included in the spanner, and
\[\sum_{i=0}^{k-1} |P_iP_{i+1}| \leq t\cdot AB\]
By the definition $AB$ should not be added to the graph because there is a path in the spanner with length at most $t\cdot AB$, which contradicts the assumption.

So assume that $AB$ is not larger than all $P_iP_{i+1}$s. Denote the largest among $P_iP_{i+1}$s by $P_{i_0}P_{i_0+1}$. Then by the assumption
\begin{align*}
\sum_{i\neq i_0} |P_iP_{i+1}| + |AB| &  \leq \sum_{i\neq i_0} |P_iP_{i+1}| + |P_{i_0}P_{i_0+1}| \\ &= \sum_{i} |P_iP_{i+1}| \leq t\cdot |AB| \leq t\cdot |P_{i_0}P_{i_0+1}|
\end{align*}
which shows that $P_{i_0}P_{i_0+1}$ can be shortcut by some smaller segments by a factor of $t$, which is impossible according to what we proved earlier in this lemma.
\end{proof}

\begin{proof}[Proof of Corollary \ref{cor:shortcut}]
Assume to the contrary that such points exist. For any non-spanner segment $P_iP_{i+1}$ there exists a path $\mathcal{P}(P_iP_{i+1})$ from $P_i$ to $P_{i+1}$ that has length at most $t\cdot |P_iP_{i+1}|$. So by replacing each non-spanner segment $P_iP_{i+1}$ by its own path $\mathcal{P}(P_iP_{i+1})$ in the shortcut path $P_0P_1\dotsc P_k$ the length of the resulting path would be
\[\sum_{P_iP_{i+1}\in S} |P_iP_{i+1}| + \sum_{P_iP_{i+1}\notin S} |\mathcal{P}(P_iP_{i+1})| \leq \sum_{P_iP_{i+1}\in S} |P_iP_{i+1}| + t\cdot\sum_{P_iP_{i+1}\notin S} |P_iP_{i+1}| \leq t\cdot |AB|\]
which shows that the new path is also a shortcut for $AB$ by a factor of $t$. But the new path only consists of the spanner segments, which is impossible by Lemma \ref{lem:shortcut} and leads to a contradiction.
\end{proof}

\begin{proof}[Proof of Lemma \ref{lem:order}]

Again, the proof goes by contradiction. Without loss of generality suppose that the projections of $P$ and $Q$ on some baseline $l$ are both between the projections of $M$ and $N$ (on the baseline). We use Lemma \ref{lem:interorder} to show that $MN$ can be shortcut by $PQ$ by a factor of $t$, i.e.
\[t\cdot |MP| + |PQ| + t\cdot |QN| \leq t\cdot |MN|\]

The idea is to move $PQ$ by a small amount with respect to its length, so that the new segment intersects $MN$, and then use Lemma \ref{lem:interorder}. We also keep track of the changes in both sides of the inequality during this movement to show the inequality holds for original points.

Let the segments $MN$ and $PQ$ intersect $AB$ at $S$ and $T$, respectively. One can move $PQ$ by vector $\overrightarrow{TS}$ in order to intersect $MN$. Let the new segment be $P'Q'$. But the projections of $P'$ and $Q'$ on the baseline may not be between $M$ and $N$ anymore. In order to preserve this, we can extend $MN$ on one side by $|\overrightarrow{TS}|$ to get a new segment $M'N'$. Extending by this amount is enough to preserve the betweenness. For example, in Figure \autoref{fig:lemorder}), we moved $PQ$ by $\overrightarrow{TS}$ to get $P'Q'$. Now $P'Q'$ intersects $MN$ (at $S$), but the projection of $P'$ on the baseline is not between the projections of $M$ and $N$ anymore. So we extend $MN$ from $M$ by $|\overrightarrow{TS}|$ to get $M'$. Now the projection of $P'$ on the baseline is between the projections of $M'$ and $N$. Before the movement the projections of $P$ and $Q$ are both between the projections of $M$ and $N$, so after movement at most one of the projections of $P'$ or $Q'$ can be outside of the projections of $M$ and $N$. So extending on one side will be sufficient.

\begin{figure}[t]
\centering
\includegraphics[width=0.6\textwidth,trim=1.5cm 4.5cm 3cm 2.5cm,clip]{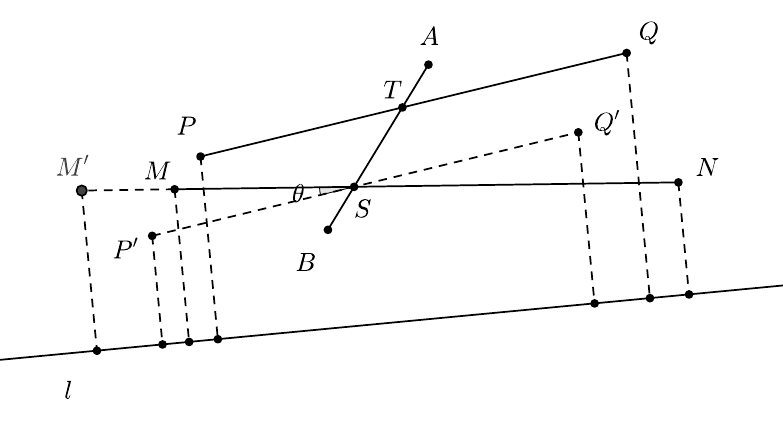}
\caption{Proof of Lemma \ref{lem:order}.}
\label{fig:lemorder}
\end{figure}

Now $P'Q'$ and $M'N'$ intersect each other and the projections of $P'$ and $Q'$ are between the projections of $M'$ and $N'$, we can use Lemma \ref{lem:interorder}. By the assumption $\theta=\frac{t'-1}{2t'}$ where $t'=(t+1)/2$, so Lemma \ref{lem:interorder} implies that,
\begin{equation}
t'\cdot |M'P'| + |P'Q'| + t'\cdot |Q'N'| \leq t'\cdot |M'N'|
\label{sni-e1}
\end{equation}
By the triangle inequality after this movement $MP$ and $NQ$ each will decrease by at most $|\overrightarrow{TS}|\leq |AB|$. So,
\begin{equation}
|M'P'| \geq |MP| - |AB|,\; |N'Q'| \geq |NQ| - |AB|
\label{sni-e2}
\end{equation}
Also length of $MN$ will increase by at most $|\overrightarrow{TS}| \leq |AB|$, so
\begin{equation}
|M'N'| \leq |MN| + |AB|
\label{sni-e3}
\end{equation}
The length of $PQ$ does not change though. Putting together \autoref{sni-e1}, \autoref{sni-e2}, and \autoref{sni-e3},
\begin{align*}
|PQ| = |P'Q'| &\leq t'\cdot(|M'N'| - |M'P'| - |N'Q'|) \\
&\leq \frac{t+1}{2}\cdot(|MN| - |MP| - |NQ| + 3|AB|) \\
%&= t'(MN - MP - NQ) + t'(3AB) \\
&\leq \frac{t+1}{2}\cdot(|MN| - |MP| - |NQ|) + \frac{t+1}{2}\cdot(\frac{t-1}{t(t+1)}|PQ|)
\end{align*}
So
\[|PQ| \leq t\cdot(|MN|-|MP|-|NQ|)\]
which is the result.
\end{proof}

\section{Many intersections with short edges}
\label{sec:smaller}
We proved in sections \ref{sec:longer} and \ref{sec:same} that, in greedy spanners, each edge has $O(1)$ crossings with edges of greater or equal length. It is natural to ask whether this holds more generally for all crossings, regardless of length. That is, is the total number of crossings for each edge bounded by a constant, depending only on $t$? In this section we will show that this is not true, by constructing a family of arrangements of points in the plane that have arbitrarily many intersections between a long edge and a set of smaller edges.

\subsection{Zig-zags}

The building block of our construction is an arrangement of points which form a zig-zag shape, as in \autoref{fig:stretch}. After running the greedy spanner algorithm on a horizontal zig-zag like this, denoted by $Z$, if $Z$ is not stretched too much along the vertical axis, the first set of edges that will be added to the graph by the greedy algorithm are actually the zig-zag edges that are drawn in \autoref{fig:stretch}. Then, depending on the shape of the zig-zag and parameter $t$, other edges may or may not be added in the future iterations. More specifically, we will show that this only depends on a parameter we call the \emph{stretch-factor} of the zig-zag.

\begin{definition}[zig-zag]
Let $Z=P_0P_1\dotsc P_k$ be a sequence of points on the Euclidean plane. We say $Z$ forms a \emph{Zig-Zag} if there exist two perpendicular vectors $\overrightarrow{\Delta x}$ and $\overrightarrow{\Delta y}$ that
\[P_i = P_0 + i\overrightarrow{\Delta x} + (i \bmod 2)\overrightarrow{\Delta y}\]
The direction of the vector $\overrightarrow{\Delta x}$ is called the \emph{direction} of the zigzag and the ratio $|\overrightarrow{\Delta y}|/|\overrightarrow{\Delta x}|$ is called the \emph{stretch factor} of the zig-zag, and is denoted by $s(Z)$. (\autoref{fig:stretch})
\label{def:zig-zag}
\end{definition}

%\begin{figure}[t]
%\centering
%\includegraphics[width=0.4\textwidth,trim=2.2cm 5cm 4cm 3.5cm,clip]{counter-zigzag}
%\caption{Zig-zag, building block of the example, and two consecutive points on it}
%\label{fig:zig-zag}
%\end{figure}

\begin{figure}[t]
\centering
\includegraphics[width=0.4\textwidth,trim=2.2cm 5cm 4cm 3.5cm,clip]{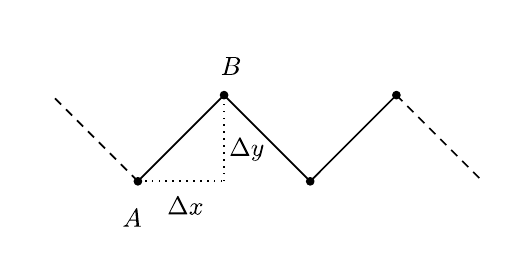}
\caption{A horizontal zig-zag, and its stretch factor $\Delta y/\Delta x$}
\label{fig:stretch}
\end{figure}

Hence a zig-zag which is more stretched toward the $\overrightarrow{\Delta y}$ vector will have a larger stretch-factor, and a zig-zag which is more stretched along the $\overrightarrow{\Delta x}$ vector will have a smaller stretch-factor.

\begin{lemma}[zig-zag spanner]
Consider a zig-zag $Z=P_0P_1\dotsc P_k$ with more than two vertices ($k>2$) in which the consecutive pairs $P_iP_{i+1}$ are connected to each other ($0\leq i < k$). For any $t > 1$, the zig-zag forms a $t$-spanner if and only if $s(Z) \leq \sqrt{t^2-1}$.
\end{lemma}
\begin{proof}
For $i< j$, the length of the path between $P_i$ and $P_j$ is
\[d_Z(P_i,P_j) = (j-i)|\overrightarrow{\Delta x} + \overrightarrow{\Delta y}|\]
while the Euclidean distance between $P_i$ and $P_j$ is
\[d(P_i,P_j) = |(j-i)\overrightarrow{\Delta x} + (j-i \bmod 2) \overrightarrow{\Delta y}|\]
The zig-zag forms a $t$-spanner if and only if $d_Z(P_i,P_j) \leq t\cdot d(P_i,P_j)$ for all $i< j$. Assume that $(j-i) \bmod 2 = 0$, this inequality turns into
\[(j-i)|\overrightarrow{\Delta x} + \overrightarrow{\Delta y}| \leq t\cdot (j-i)|\overrightarrow{\Delta x}|\]
which is equivalent to $s(Z) \leq \sqrt{t^2-1}$. So this is a necessary condition, and it can be shown that it is a sufficient condition too. Because assuming $s(Z) \leq \sqrt{t^2-1}$, in a similar way,
\[(j-i)|\overrightarrow{\Delta x} + \overrightarrow{\Delta y}| \leq t\cdot (j-i)|\overrightarrow{\Delta x}|\]
The left side of the inequality is $d_Z(P_i,P_j)$ and the right side is no more than $t\cdot d(P_i,P_j)$ because it is missing the term $(j-i \bmod 2) \overrightarrow{\Delta y}$, so $d_Z(P_i, P_j) \leq t\cdot d(P_i, P_j)$.
\end{proof}

\subsection{Introducing the arrangement}
Now we introduce the arrangement. Consider two horizontal zig-zags $U$ on the top and $B$ on the bottom which are connected together using a middle zig-zag $M$ (\autoref{fig:counter}). $U$ is colored by green, $B$ is colored by blue, and $M$ is colored by red. So there are four rows of points and three zig-zags $U$, $M$, and $B$, which connect these points together. $U$ and $M$ share the second row, while $M$ and $B$ share the third row. The first row is only included in $U$, and the last row is only included in $B$. For now, suppose that there are enough points in each row. Later we will see that if the number of points is larger than a specific amount, then a large edge appears at some point in the greedy algorithm, intersecting many edges in between.

\begin{figure}[t]
\centering
\includegraphics[width=0.6\textwidth,trim=1cm 3cm 1cm 3cm,clip]{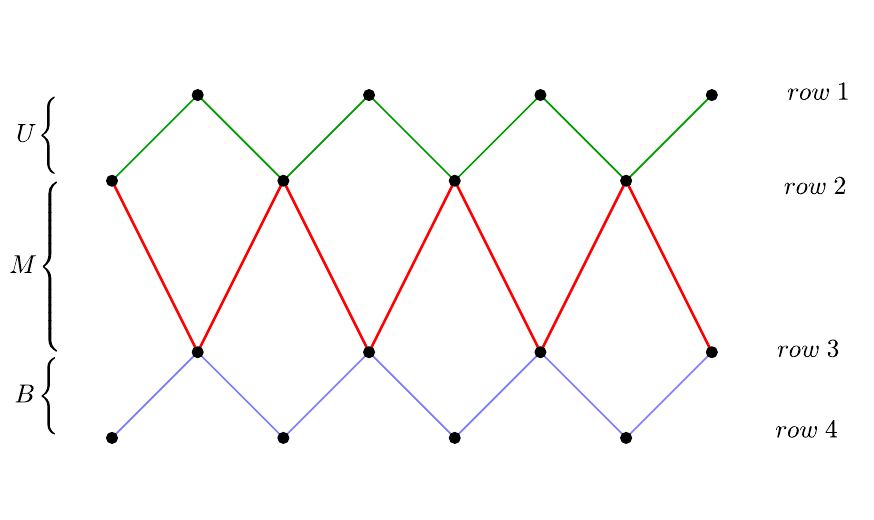}
\caption{Example with more than constant intersections with smaller edges}
\label{fig:counter}
\end{figure}

All of the zig-zags $U$, $M$, and $B$ can have arbitrary stretch-factors as we can move the rows up or down to adjust the stretch-factor of each zig-zag independently. So assume that $s(U)=s(B)=\sqrt{t^2-1}$ and $s(M)=\sqrt{(t+\delta)^2-1}$, for some small positive $\delta$ which will be specified later. In other words, $U$ and $B$ are the most stretched zig-zags that form a $t$-spanner and $M$ is a slightly more stretched zig-zag, which is not a $t$-spanner by itself anymore.

With this choice of stretch-factors, it is not hard to see, by the Pythagorean theorem, that the length of the zig-zag path between two points on $U$, say $a$ and $b$, is exactly $t\cdot |x_a-x_b|$. And the length of the path between two points on $B$ is also the same expression. But in a similar way, the length of the zig-zag path between two points on $M$ would be slightly more, $(t+\delta) \cdot |x_a-x_b|$.

\subsection{Simulating the greedy algorithm on the arrangement}
For an appropriate choice of $t$ (one causing the angles of all zig-zags to lies strictly between $60^\circ$ and $120^\circ$), the greedy spanner algorithm will first add the zig-zag edges in $U$, $B$, and $M$, as they are the closest pairs of vertices. According to the chosen stretch-factors, no edges will be added to $U$ and $B$ in the future. For example, the horizontal dashed blue edges in \autoref{fig:counter-vert} will not be added as the endpoints of these segments both belong to $U$ or $B$, which are $t$-spanners by themselves. So any potential edge must be between $U$ and $B$.

\begin{figure}[t]
\centering
\includegraphics[width=0.6\textwidth,trim=1cm 3cm 1cm 3cm,clip]{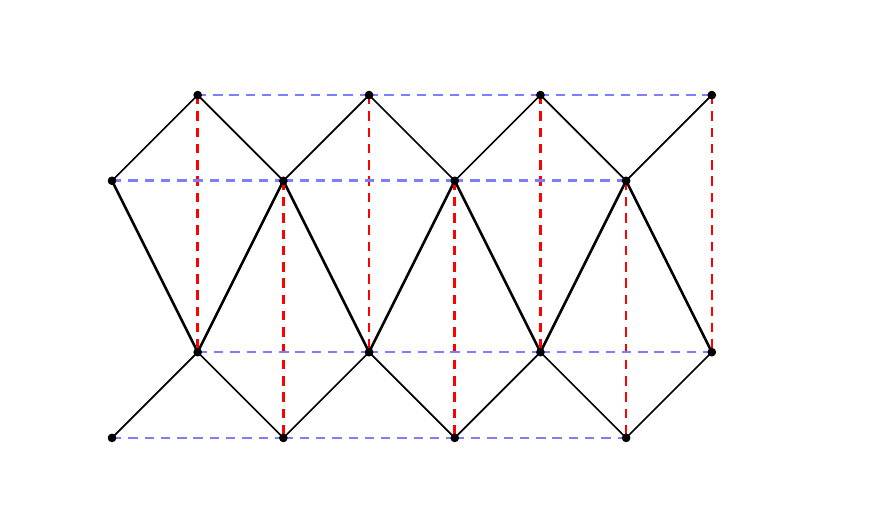}
\caption{Vertical dashed segments are included in the graph but not horizontal ones.}
\label{fig:counter-vert}
\end{figure}

The next set of edges that may be added by the algorithm are the vertical edges between rows 1 and 3, and 2 and 4 (red dashed segments in \autoref{fig:counter-vert}). These are the closest pairs across $U$ and $B$ which are not connected, so they will be included first. The edges between rows 1 and 4 which connect the points in consecutive columns (dashed blue segments in \autoref{fig:counter-cons}) may also be added in the next iteration, depending on how small the value of $t$ is, but we will see that they do not affect the length of the shortest paths between pairs of points in $U$ and $B$ that much.

\begin{figure}[t]
\centering
\includegraphics[width=0.6\textwidth,trim=1cm 3cm 1cm 3cm,clip]{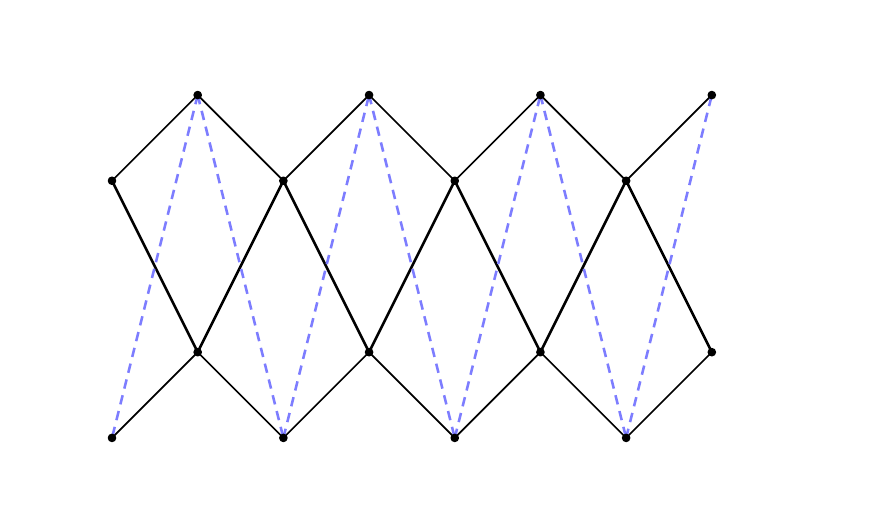}
\caption{The big dashed zig-zag might be included in the graph or might not.}
\label{fig:counter-cons}
\end{figure}

\subsection{Sufficiency of small edges for close pairs}
Now we claim that the edges we found until now are the only local, i.e. small, edges between these points, and the next edge that is going to be added by the greedy algorithm, would be a large one which intersects many of the zig-zag edges in $M$. The greedy algorithm may stop here and do not add any edges, but we will prove later that this is not possible. We are not going to address this issue in this section.

Intuitively, one can use edges in $U$ and $B$, and only one edge in $M$ to build a path from any point in $U$ to any point in $B$ (see \autoref{fig:counter-path}). Again, intuitively, zig-zags are defined in a way that the length of this path is more than $t\cdot|x_u-x_b|$ by a small constant. But when $u$ and $b$ are not far away $|x_u-x_b|$ is much less than $d(u,b)$ and hence the length of the path is no more than $t\cdot d(u,b)$. On the other hand, when $u$ and $b$ are far away, $|x_u-x_b|$ is closer than any constant to $d(u,b)$ (because here $|y_u-y_b|$ is bounded), hence the length of the path becomes more than $t\cdot d(u,b)$ and a long edge appears.

In order to prove this formally, as stated above, any potential edge must be between $U$ and $B$. So let $u\in U$ and $b\in B$ be two arbitrary points in the top and the bottom zig-zags, respectively. Also assume that $u$ is the $i$-th point in $U$ ($i=0,1,\dotsc$), and $b$ is the $j$-th point in $B$ ($j=0,1,\dotsc$), counting from left (\autoref{fig:counter-path}).

We assume that no edges other than the ones we stated above have been added so far, and we compute the length of a path we propose between $u$ and $b$ that uses these edges and we show that it is less than $t\cdot d(u,b)$ if $d(u,b)$ is not very large. In this way, we prove that the next edge which is going to be added would be a large one.

Without loss of generality, assume that $i\leq j$. Consider a path that uses zig-zag edges of $U$ and $B$ and only one of the edges in $M$ to reach from $u$ to $b$. Denote this path by $P(u,b)$. Such a path is drawn by a red dashed line for two sample points in \autoref{fig:counter-path}. Clearly, we do not use any edge twice and we only use zig-zag edges in $U$, $B$, or $M$.

\begin{figure}[t]
\centering
\includegraphics[width=0.5\textwidth,trim=2cm 3.5cm 2.5cm 3.5cm,clip]{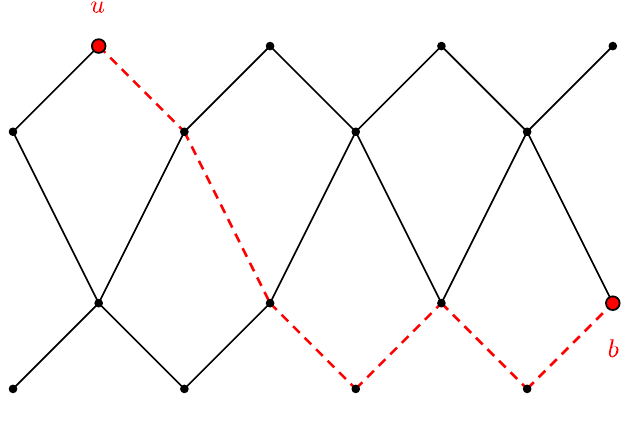}
\caption{$P(u,b)$, which uses some of the edges in $U$ and $B$ and only one edge in $M$. Here $i=1$ and $j=7$.}
\label{fig:counter-path}
\end{figure}

We will show that $|P(u,b)|$, the length of the red path, is not more than $t\cdot d(u,b)$ when $d(u,b)$ is not very large. By the definition, $P(u,b)$ uses $j-i-1$ edges of $U$ and $B$, and one edge in $M$, so
\begin{equation}
|P(u,b)| = (j-i-1)l + l'
\label{counter-len-path}
\end{equation}
where $l$ is the edge length in $U$ (and $B$), and $l'$ is the edge length in $M$. On the other side, the distance along the $x$-axis between $u$ and $b$ is $(i-j)\Delta x$, where $\Delta x$ is defined in Definition \ref{def:zig-zag}. The distance along the $y$-axis between $u$ and $b$ is at least the height of the zig-zag $M$, which is by the definition $s(M)\Delta x$. This distance can be strictly more than $s(M)\Delta x$ when $u$ is in the first row or $b$ is in the last row. So,
\begin{equation}
d(u,b) \geq \sqrt{(j-i)^2(\Delta x)^2 + s(M)^2(\Delta x)^2} = \sqrt{(j-i)^2 + (t+\delta)^2 - 1}\Delta x
\label{counter-dist}
\end{equation}
In order to show $|P(u,b)|\leq t\cdot d(u,b)$, we use \autoref{counter-len-path} and \autoref{counter-dist} to show $|P(u,b)|^2 - t^2\cdot d(u,b)^2$ is non-positive,
\begin{align*}
|P(u,b)|^2 - t^2\cdot d(u,b)^2 &\leq \left[(j-i-1)l + l')\right]^2 - \left[(j-i)^2 + (t+\delta)^2 - 1\right](t\Delta x)^2 \\
&= \left[(j-i) + (l'/l -1))\right]^2l^2 - \left[(j-i)^2 + (t+\delta)^2 - 1\right]l^2 \\
&= \left[2(j-i)(l'/l-1) + (l'/l-1)^2 - (t+\delta)^2 + 1\right]l^2
\end{align*}
We used $t\Delta x = l$ in the first equality. Now by putting $l'/l = \frac{t+\delta}{t}$, when $j-i \leq t(t^2-1)/(2\delta)$,
\begin{align*}
|P(u,b)|^2 - t^2\cdot d(u,b)^2 
&\leq \left[2(j-i)\frac{\delta}{t} + (\frac{\delta}{t})^2 - (t+\delta)^2 + 1\right]l^2 \\
&\leq \left[(t^2-1) + \delta^2 - (t+\delta)^2 + 1\right]l^2 \leq 0
\end{align*}
So no edge is required between $u$ and $b$ and if there is any edge between them, it must be the case that $j-i > t(t^2-1)/(2\delta)$. On the other side, the edge $(u,b)$, if exists, will intersect at least $j-i-2$ of the zig-zag edges which separate $u$ and $b$. So one can choose $\delta$ to be sufficiently small to increase the number of intersections.

\subsection{Existence of a large edge}
Now we address the issue we mentioned earlier, that the greedy algorithm may stop after adding the small edges we discussed in section 3.3 and never add any large edges. We need to prove the existence of such a large edge to complete the proof.

Again, let $u$ be the $i$-th point in $U$ and $b$ be the $j$-th point in $B$, counting from left. We will show that when $j-i$ is large enough an edge is required between $u$ and $b$. None of the edges that we mentioned so far connects two points whose $x$-distance is more than $\Delta x$. So the shortest path between $u$ and $b$, denoted by $P^*(u,b)$, needs at least $j-i$ edges to reach from $u$ to $b$. At least one of these edges should be across $U$ and $B$, hence having a length at least $l'$. The other edges have lengths of at least $l$, as it is the smallest edge in the graph. Thus
\begin{equation}
|P^*(u,b)| \geq (j-i-1)l + l'
\label{counter-len-path-2}
\end{equation}
Again, the $x$-distance of $u$ and $b$ is $(i-j)\Delta x$, and the $y$-distance of them is at most the height of the whole figure, which is the sum of the height of the three zig-zags, $(s(U) + s(M) + s(B))\Delta x$. So,
\begin{align}
\begin{split}
d(u,b) &\leq \sqrt{(j-i)^2(\Delta x)^2 + (s(U)+s(M)+s(B))^2(\Delta x)^2} \\
&\leq\sqrt{(j-i)^2 + (3s(M))^2}\Delta x
=\sqrt{(j-i)^2 + 9(t+\delta)^2 - 9}\Delta x
\end{split}
\label{counter-dist-2}
\end{align}
The second inequality follows from the fact that $s(M)$ is the maximum among $s(U)$, $s(M)$, and $s(B)$. Similarly, we use \autoref{counter-len-path-2} and \autoref{counter-dist-2} to show that $|P(u,b)|^2 - t^2\cdot d(u,b)^2$ is positive,
\begin{align*}
|P^*(u,b)|^2 - t^2\cdot d(u,b)^2 
&\geq \left[2(j-i)\frac{\delta}{t} + (\frac{\delta}{t})^2 - 9((t+\delta)^2 - 1)\right]l^2
\end{align*}
When $j-i\geq 9t((t+\delta)^2-1)/(2\delta)$,
\[|P^*(u,b)|^2 - t^2\cdot d(u,b)^2 \geq \left[9((t+\delta)^2-1) + (\frac{\delta}{t})^2 - 9((t+\delta)^2 - 1)\right]l^2 > 0\]
Hence the result.

\begin{theorem}
\label{thm:many-crossings}
For some values of $t$, there is no constant bound (depending only on $t$) on the number of crossings between an edge of a greedy $t$-spanner and other smaller edges.
\label{th:smaller}
\end{theorem}
\begin{proof}
This follows from the existence of the example above.
\end{proof}

\end{document}